\newif\ifHasAppendix
\documentclass[11pt,a4paper]{article}

\usepackage[margin=1in]{geometry}
\usepackage[T1]{fontenc}
\usepackage[utf8]{inputenc}
\usepackage{lmodern}
\usepackage{microtype}

\usepackage[colorlinks,allcolors=blue!60!black]{hyperref}

\usepackage{amsmath,amsthm,amssymb}
\usepackage{mathtools}
\usepackage[square, numbers, comma, sort&compress]{natbib}
\usepackage{cleveref}
\usepackage{algorithm}
\usepackage{algorithmic}
\usepackage{natbib}

\usepackage{tikz}

\usepackage{authblk}

\usepackage{latexsym}
\usepackage{cleveref}
\usepackage{xargs}
\usepackage{xcolor}
\usepackage{amssymb}
\usepackage{todonotes}
\usepackage{makecell}
\usepackage{tablefootnote}
\usepackage{mathtools}
\usepackage{booktabs}

\usepackage{thmtools}
\usepackage{thm-restate}
\newtheorem{theorem}{Theorem}
\newtheorem{lemma}[theorem]{Lemma}

\theoremstyle{definition}
\newtheorem{definition}{Definition}

\newtheorem{example}{Example}
\crefname{theorem}{Theorem}{Theorems}
\Crefname{theorem}{Thm}{Thms}
\crefname{lemma}{Lemma}{Lemmas}
\Crefname{lemma}{Lem}{Lems}
\newtheorem{corollary}[theorem]{Corollary}
\newtheorem{proposition}[theorem]{Proposition}

\crefname{proposition}{Proposition}{Propositions}
\Crefname{proposition}{Prop}{Props}

\newtheorem{observation}[theorem]{Observation}
\newtheorem{problem}{Problem}
\newtheorem{remark}{Remark}

\title{Envy-free Allocations with Individual Payments}

\author[1]{Robert Bredereck}
\author[1]{Eva Deltl}
\author[2]{Tanmay Inamdar}
\author[2]{Pallavi Jain}
\author[2]{Pranjal Pandey}

\affil[1]{TU Clausthal, Germany}
\affil[2]{IIT Jodhpur}

\date{} 

\begin{document}
\maketitle

\newcommandx{\problemdef}[7][3=Input,5=Question]{
  \begingroup\par\noindent\nopagebreak[4]
  \begin{problem}\label{prob:#2}\vspace{-0.5em}\hspace{-0.5em}\colorbox{white}{\begin{tabular}{@{}l@{}}\textsc{#1}  \textsc{#7}\end{tabular}}\nopagebreak[4]\end{problem}\nopagebreak[4]
  \par\noindent\hangindent=\parindent\textbf{#3}:  #4\nopagebreak[4]
  \par\noindent\hangindent=\parindent\textbf{#5}:  #6
  \par\medskip
  \endgroup
}

\newcommand{\pjmargin}[1]{\todo[color=red!20!white]{\textsc{PJ:} #1}}
\newcommand{\timargin}[1]{\todo[color=red!20!white, size=tiny]{\textsc{TI:} #1}}
\newcommand{\pjinline}[1]{\todo[inline, color=red!20!white]{\textsc{PJ:} #1}}
\newcommand{\edinline}[1]{\todo[inline, color=blue!20!white]{\textsc{ED:} #1}}
\newcommand{\edmargin}[1]{\todo[color=blue!20!white]{\textsc{ED:} #1}}
\newcommand{\ppinline}[1]{\todo[inline, color=green!20!white]{\textsc{PP:} #1}}
\newcommand{\rbinline}[1]{\todo[inline, color=yellow!20!white]{\textsc{Rb:} #1}}
\newcommand{\rbmargin}[1]{\todo[color=yellow!20!white]{\textsc{RB:} #1}}

\newcommand{\N}{\mathbb{N}} 

\begin{abstract}
When an envy-free allocation of indivisible goods does not exist, monetary transfers can restore envy-freeness. Existing work on fair division with subsidies, however, typically assumes that these payments are provided by an external source, an assumption that may be unrealistic in many applications. We address this limitation by allowing only monetary transfers between agents, with each agent’s payments constrained by their individual budget.
We show that while it is polynomial-time tractable to determine whether a given allocation can be made envy-free by payments under individual budgets, the general problem of computing such an allocation from scratch is NP-hard, even when agents have relatively large budgets.
Motivated by this intractability, we conduct a parameterized complexity analysis, establishing fixed-parameter tractability with respect to the number of goods or to the joint parameter number of agents and good types, and we provide efficient algorithms in several special cases. For explicitly listed items, our type-based algorithm answers the envy-freeness part of an open question of \citet{DBLP:conf/ijcai/0004R23}.
\end{abstract}

\newcommand{\appendixsection}[1]{%
  \gappto{\appendixProofText}{\section{#1}}
}
\newcommand{\appendixsubsection}[1]{%
  \gappto{\appendixProofText}{\subsection{Additional Material for Section~\ref{#1}}\label{app:#1}}
}
\newcommand{\toappendix}[1]{%
  \gappto{\appendixProofText}{{#1}}
}
\newcommand{\appendixproof}[2]{%
  \gappto{\appendixProofText}{\subsection{Proof of \Cref{#1}}\label{proof:#1}#2}
}
         
\newcommand{\BibTeX}{\rm B\kern-.05em{\sc i\kern-.025em b}\kern-.08em\TeX}


\maketitle 


\newcommand{\localBudget}{\textsc{Envy-free allocations with individual payments}}
\newcommand{\localBudgetAbrv}{\textsc{EFA-IP}}

\newcommand{\initialfair}{\textsc{Envy-free Extension with individual payments}}\newcommand{\initialfairAbrv}{\textsc{EFE-IP}}

\newcommand{\R}{\ensuremath{\mathbb{R}}}
\newcommand{\Z}{\ensuremath{\mathbb{Z}}}

\newcommand{\allocation}{\mathcal{X}}

\section{Introduction}

Fair division addresses the problem of allocating indivisible resources to agents in a way that satisfies fairness and efficiency criteria \cite{DEMKO1988145,Brams_Taylor_1996}. A central notion in this domain is \emph{envy-freeness}, which requires that no agent prefer another agent’s allocation to their own. However, for indivisible goods, envy-freeness is often unattainable, e.g., when one good must be allocated between two agents.

A common approach to mitigate envy is to allow monetary compensations, including subsidies, donations, or side payments.
This leads to hybrid models that combine allocations of indivisible goods with monetary transfers, which can, in principle, resolve envy-freeness when it is otherwise impossible.
However, most models assume either a central authority that can redistribute funds or access external money (for example, via subsidies or charity), which is not always realistic.
In many real-world scenarios, transfers are made between agents and should be budget-balanced, meaning the sum of payments must be zero \cite{budget-feasible}.

Even under budget-balance, existing models such as the pricing framework of \citet{AzizPricing} typically assume agents can make unrestricted payments. 
This overlooks a key constraint present in most real-world settings: agents often have unequal initial capital and therefore cannot make arbitrary monetary transfers. For instance, in inheritance division or in the division of jointly owned assets,
individuals may differ significantly in what they can afford. Ignoring such constraints can lead to allocations that are fair in theory but infeasible in practice.

To address this, we introduce a fair division framework that incorporates \emph{agent-specific budget constraints}, modeled as hard upper bounds on the payments each agent can make. These constraints directly influence both feasibility of transfers and the structure of envy, yielding a more realistic model of fair allocation that better reflects real-world limitations.

\vspace{.5em}
\textbf{ Related work.}
The literature on fair division has explored various mechanisms to restore fairness when envy-freeness cannot be achieved via allocation alone. 

One approach is to relax the notion of envy-freeness. Prominent relaxations include EF1 (envy-freeness up to one good) \cite{Ef1} and EFx (envy-freeness up to any good) \cite{EFx}; see \citet{AMANATIDIS2023103965} for a survey. 
However, when monetary transfers are permitted, even if constrained, the pursuit of exact envy-freeness is both natural and practically justified. We therefore do not study these relaxations.

Another line of work focuses on modifying the instance, for example by adding or removing items to enable a fair allocation \cite{aziz2016control,DBLP:journals/algorithmica/DornHS21,boehmer2024multivariate,ijcai2025p417}.

Our focus is on models that introduce divisible resources, such as money, to compensate agents and eliminate envy. 
This idea is rooted in the classical rent division problem, where $n$ agents are each assigned one of $n$ indivisible rooms and share the rent in a way that guarantees envy-freeness \cite{money2,haake2002bidding,money}.
More generally, when the number of items exceeds the number of agents, monetary transfers or subsidies can help restore envy-freeness. A common goal is to find the minimal set of payments that eliminates all envy.
\citet{DBLP:conf/sagt/HalpernS19} introduce a framework for quantifying the price of envy-freeness, analyzing how much subsidy is needed to eliminate envy. \citet{SubsidiesGoko} take a mechanism-design perspective and analyze the subsidies needed to ensure truthfulness, envy-freeness, and utilitarian optimality.

These monetary models can be distinguished along three axes: (i) allocation structure (one-to-one vs.\ one-to-many), (ii) whether transfers must be budget-balanced or may rely on external subsidies, and (iii) whether transfers are unrestricted or subject to individual budgets. Most recent work has focused on the budget-balanced but unrestricted case in the one-to-many setting, designing transfer rules to guarantee envy-freeness while maximizing welfare or fairness \cite{AzizPricing,aziz2025weighted,twobirds}. These models typically assume unconstrained internal redistribution, meaning agents can transfer money freely among themselves as long as the overall budget is balanced. More generally, \citet{oneDollar} showed that for additive valuations, scaling item values so that no marginal contribution exceeds one dollar, an envy-free allocation always exists with subsidies of at most one dollar per agent. Recently, \citet{lu2026optimalsubsidyboundsgoods} extended this tight bound to mixed settings with
goods and chores. Subsidies have also been studied for agents with unequal entitlements:
\citet{aziz2025weighted} and \citet{kleinelmalem2025weighted} consider
\emph{weighted} envy-freeness with subsidies and establish bounds on the
subsidy required. As in the unweighted case, these models rely on an external
source of money, whereas we require budget-balanced transfers bounded by
individual budgets.

The idea of individual budget constraints is well-studied in the auction
and mechanism design literature
\cite{budgetconstraint,budgets,borgs2005multi,lerner2013dictatorial},
where a bidder's budget typically represents the maximum price they can pay
to the seller. For instance, \citet{borgs2005multi} show that no
non-trivial truthful multi-unit auction with budget-constrained bidders can
allocate all items, and \citet{lerner2013dictatorial} show that in
constrained combinatorial auctions, incentive compatibility together with
Pareto efficiency essentially forces dictatorial mechanisms.
There, envy arises only when another bidder obtains a good at a price the
envious bidder could have matched. In contrast, our setting involves direct
transfers between agents with no seller collecting revenue, where both
the direction and size of payments can themselves generate envy.

Closer to our setting, 
\citet{Procaccia_Velez_Yu_2018} introduced agent-specific budget constraints in the context of envy-free rent division, but restricted attention to one-to-one allocations (each agent receives exactly one room). 
In the one-to-many setting, \citet{budget_chores} proposed a budget-constrained model for allocating indivisible chores and analyze how goods and chores differ.
Similarly, \citet{Barman_Khan_Shyam_Sreenivas_2023} 
studied a setting in which each good has an agent-specific value and size, and every agent has a budget constraint that limits the total size of the goods she can receive.
However, their notion of budget differs from ours. We interpret an agent’s budget as the maximum payment it can make and incorporate this directly into our envy relations, whereas they assign weights to chores and define budget as an upper bound on the total weight an agent can receive.

Finally, \citet{bilo_et_al:LIPIcs.ESA.2024.26} also recognized that assuming access to external money is often unrealistic in real-world settings. They proposed an interesting approach in which certain goods are sold, and the compensation is used to achieve envy-freeness through monetary transfers.

Our contribution complements these works by introducing the \localBudget{} model (\localBudgetAbrv{}, defined in Section~\ref{sec:prelims}), allowing us to analyze one-to-many allocations under budget-balance with agent-specific budgets, and filling a gap in the existing taxonomy of fair division with money.

\vspace{.5em}
\textbf{ Our contributions.}
We study the computational complexity of \localBudgetAbrv{} (defined in Section~\ref{sec:prelims}), which asks whether there exists an envy-free allocation with payments such that transfers are budget-balanced and respect each agent’s budget constraint. We focus on agents with additive, non-negative valuations. 
Figure~\ref{table:params} summarizes our results.

\begin{figure*}[t]
\centering

\resizebox{\textwidth}{!}{%
\begin{tikzpicture}[
  every node/.style={
    font=\fontsize{8.5pt}{9.5pt}\selectfont
  },
  hdr/.style={
    draw=gray!70,
    fill=gray!20,
    rounded corners=2pt,
    align=center,
    minimum height=0.48cm,
    font=\fontsize{8.5pt}{9.5pt}\selectfont\bfseries,
    inner sep=2pt
  },
  pnode/.style={
    draw=gray!55,
    fill=white!9,
    rounded corners=2pt,
    align=center,
    inner sep=2pt
  },
  gbox/.style={
    draw=green!55!black,
    fill=green!9,
    rounded corners=2pt,
    align=center,
    inner sep=2pt
  },
  rbox/.style={
    draw=red!65!black,
    fill=red!5,
    rounded corners=2pt,
    align=center,
    inner sep=2pt
  },
  bbox/.style={
    draw=orange!55!black,
    fill=orange!8,
    rounded corners=2pt,
    align=center,
    inner sep=2pt
  },
]


\def\xP{1.1}
\def\wP{2.2cm}

\def\xC{4.2}
\def\wC{3.7cm}

\def\xI{8.5}
\def\wI{4.60cm}

\def\xG{14.3}
\def\wG{6.70cm}

\def\xSpanThreeFour{11.925}
\def\wSpanThreeFour{11.45cm}


\def\hHdr{0.48cm}
\def\hN{1.00cm}
\def\hM{0.55cm}
\def\hNL{0.50cm}
\def\hNV{0.50cm}
\def\hTL{0.75cm}
\def\hBV{0.55cm}


\def\yHdr{0}
\def\yN{-0.79}
\def\yM{-1.57}
\def\yNL{-2.12}
\def\yNV{-2.68}
\def\yTL{-3.365}
\def\yBV{-4.075}


\node[
  hdr,
  text width=\wP
] at (\xP,\yHdr)
{Parameter};

\node[
  hdr,
  text width=\wC
] at (\xC,\yHdr)
{Classification};

\node[
  hdr,
  text width=\wI
] at (\xI,\yHdr)
{Identical valuations};

\node[
  hdr,
  text width=\wG
] at (\xG,\yHdr)
{General valuations};


\node[
  pnode,
  text width=\wP,
  minimum height=\hN
] at (\xP,\yN)
{$n$};

\node[
  rbox,
  text width=\wC,
  minimum height=\hN
] at (\xC,\yN)
{%
  \textbf{para-NP-hard} (binary)\\[-1pt]
  \textbf{$W[1]$-hard} (unary)\\[-1pt]
  [Thm.~\ref{thm:npw}]
};

\node[
  pnode,
  text width=\wI,
  minimum height=\hN
] at (\xI,\yN)
{--};

\node[
  pnode,
  text width=\wG,
  minimum height=\hN
] at (\xG,\yN)
{--};


\node[
  pnode,
  text width=\wP,
  minimum height=\hM
] at (\xP,\yM)
{$m$};

\node[
  gbox,
  text width=\wC,
  minimum height=\hM
] at (\xC,\yM)
{\textbf{FPT}};

\node[
  gbox,
  text width=\wI,
  minimum height=\hM
] at (\xI,\yM)
{%
  $\mathrm{poly}(n,m)\cdot 2^m$
  [Thm.~\ref{thm:fptm}]
};

\node[
  gbox,
  text width=\wG,
  minimum height=\hM
] at (\xG,\yM)
{%
  $\mathrm{poly}(n,m)\cdot m^{2m}$
  [Thm.~\ref{thm:m}]
};


\node[
  pnode,
  text width=\wP,
  minimum height=\hNL
] at (\xP,\yNL)
{$n+\lambda$};

\node[
  gbox,
  text width=\wC,
  minimum height=\hNL
] at (\xC,\yNL)
{\textbf{FPT}};

\node[
  gbox,
  text width=\wSpanThreeFour,
  minimum height=\hNL
] at (\xSpanThreeFour,\yNL)
{%
  \emph{MILP-solvable}
  [Thm.~\ref{thm:fptnt}]
};


\node[
  pnode,
  text width=\wP,
  minimum height=\hNV
] at (\xP,\yNV)
{$n+v_{\max}$};

\node[
  gbox,
  text width=\wC,
  minimum height=\hNV
] at (\xC,\yNV)
{\textbf{FPT}};

\node[
  gbox,
  text width=\wSpanThreeFour,
  minimum height=\hNV
] at (\xSpanThreeFour,\yNV)
{%
  \emph{MILP-solvable}
  [Cor.~\ref{thm:value}]
};


\node[
  pnode,
  text width=\wP,
  minimum height=\hTL
] at (\xP,\yTL)
{$\tau+\lambda$};

\node[
  bbox,
  text width=\wC,
  minimum height=\hTL
] at (\xC,\yTL)
{\textbf{XP}\footnotemark};

\node[
  bbox,
  text width=\wI,
  minimum height=\hTL
] at (\xI,\yTL)
{%
  $\mathrm{poly}(n,m)
  \big(1+\tfrac{m}{\lambda}\big)^{2\lambda}$
  [Thm.~\ref{thm:types}]
};

\node[
  bbox,
  text width=\wG,
  minimum height=\hTL
] at (\xG,\yTL)
{%
  $\mathrm{poly}(n,m)
  (1+nb_{\max})^{\tau}
  \big(1+\tfrac{m}{\lambda}\big)^{\lambda(\tau+1)}$
  [Thm.~\ref{thm:more-types}]
};


\node[
  pnode,
  text width=\wP,
  minimum height=\hBV
] at (\xP,\yBV)
{$b_{\max}+v_{\max}$};

\node[
  rbox,
  text width=\wC,
  minimum height=\hBV
] at (\xC,\yBV)
{%
  \textbf{para-NP-hard}
  [Thm.~\ref{thm:npw}]
};

\node[
  pnode,
  text width=\wI,
  minimum height=\hBV
] at (\xI,\yBV)
{--};

\node[
  pnode,
  text width=\wG,
  minimum height=\hBV
] at (\xG,\yBV)
{--};

\end{tikzpicture}%
}

\caption{Overview of our results for \localBudgetAbrv{}. Here, $n$, $m$, $\tau$, and $\lambda$
denote the numbers of agents, goods, agent types, and good types,
respectively; $b_{\max}$ and $v_{\max}$ are the maximum budget and good valuation.}

\label{table:params}

\end{figure*}

\footnotetext{For general valuations, payments between groups are assumed to be integral; the XP classification assumes unary-encoded budgets.}

A key difference from other monetary models is that with individual budget constraints, an envy-free allocation with payments may fail to exist. This is clear in the case where all agents have zero budget, which reduces to finding a classical envy-free allocation (Thm~\ref{thm:npw}), known not to always exist.

On the positive side, we show that checking whether a given allocation
can be made envy-free using individual payments is solvable in polynomial time
by modeling it as a linear program (LP) (Thm~\ref{thm:lp}).
However, aligning with the hardness of finding envy-free allocations in the classical setting, computing such an allocation from scratch is NP-hard. More generally, we characterize the role of the budget through a sharp complexity threshold: the problem is polynomial-time solvable when every agent has budget $b_a=\max_{g\in G} v_a(g)$ (Lem~\ref{subsidy}), but becomes NP-hard when the budget is decreased by one to $b_a=\max_{g\in G} v_a(g)-1$ (Thm~\ref{thm:np}).

Motivated by this computational intractability, which already holds when agents have binary valuations, we initiate a parameterized complexity analysis of our problem.

We present an algorithm that achieves fixed-parameter tractability (FPT) with respect to the number of goods under both identical and general valuations (Thms~\ref{thm:fptm} and~\ref{thm:m}).
Additionally, we show that the problem is FPT when jointly parameterized by the number of agents $n$ and the number of good types $\lambda$ (Thm~\ref{thm:fptnt}), while $W[1]$-hard for only $n$ (Thm~\ref{thm:npw}).

Finally, assuming integral payments that are polynomially bounded in the input size, we give an XP algorithm parameterized by the numbers of agent types~$\tau$ and good types~$\lambda$. For explicitly listed goods, this answers the envy-freeness part of an open question posed by~\citet{DBLP:conf/ijcai/0004R23}.

\section{Preliminaries}
\label{sec:prelims}

For~$n \in \N$,\footnote{We use the convention $\N=\{0,1,2,\dots\}$.} with $n > 0$
let~$[n] \coloneq \{1, 2, \dots, n\}$.

Let $G = \{g_1, \dots, g_m\}$ be the set of \emph{goods} and let $A = \{a_1, \dots, a_n\}$ denote the set of \emph{agents}. 
Each agent $a \in A$ has a budget $b_a \in \N$ and a valuation function $v_a \colon \mathcal{P}(G) \to \N$. With slight abuse of notation, we write $v_a(g) = v_a(\{g\})$. We assume valuations are \emph{additive}, i.e., $v_a(S) \coloneqq \sum_{g \in S} v_a(g)$ for all $S \subseteq G$.
We call a valuation $v_a$ \emph{binary} if for all $g \in G$ we have $v_a(g) \in \{0,1\}$, and say that $a$ approves good $g$ if $v_a(g) = 1$.
Valuations $v_a$ and $v_{a'}$ are \emph{identical} if $v_a = v_{a'}$.
We group agents and goods by identical valuations.
Let $A = A_1 \uplus \dots \uplus A_\tau$, where $a, a' \in A_t$ if and only if $v_a = v_{a'}$; agents in $A_t$ are said to be of \emph{type}~$t$.
Similarly, let $G = G_1 \uplus \dots \uplus G_\lambda$, where $g, g' \in G_\ell$ if and only if $v_a(g) = v_a(g')$ for all $a \in A$; goods in $G_\ell$ are of \emph{type}~$\ell$.

An \emph{allocation} is a partition $\allocation = (X_a)_{a\in A}$ of $G$ into $n$ disjoint bundles,  one for each agent. Note that we consider complete allocations, and the bundles $X_i$ are allowed to be empty.
Agent $a$ \emph{(initially) envies} agent $a'$ if $v_a(X_a) < v_a(X_{a'})$.

An \emph{allocation with payments} is a pair $(\allocation, \mathbf{p})$, where $\allocation$ is an allocation and $\mathbf{p} \in \mathbb{R}^n$ is a vector of payments. The \emph{effective utility} of agent $a$ is
$u_a(\allocation, \mathbf{p}) = v_a(X_a) + p_a$,
where $p_a$ is the \emph{net payment} made (if negative) or received (if positive) by agent $a$. 
That is, agents have \emph{quasi-linear} utilities, meaning their utility is linear in money and equals the value of the allocated bundle plus the net payment. 
The \emph{utilitarian welfare} of an allocation with payments $(\allocation, \mathbf{p})$ is defined as the sum of effective utilities across all agents, i.e., $\sum_{a \in A} u_a(\allocation, \mathbf{p})$.
Given two agents $a, a' \in A$, we say that $a$ \emph{envies} $a'$ under $(\allocation, \mathbf{p})$ if $v_a(X_a) + p_a < v_a(X_{a'}) + p_{a'}$.

An allocation with payments is \emph{envy-free} if no agent envies another under $(\allocation, \mathbf{p})$.
We say that an allocation $\allocation$ is \emph{envy-freeable} if there exists a payment vector $\mathbf{p}$ such that $(\allocation, \mathbf{p})$ is envy-free. We then say that $\mathbf{p}$ \emph{eliminates envy}.

A payment vector $\mathbf{p}$ is \emph{feasible} if $p_a \geq -b_a$ for all $a \in A$, i.e., no agent's net payment exceeds their budget, and \emph{budget-balanced} if $\sum_{a \in A} p_a = 0$, i.e., total payments sum to zero, meaning all transfers occur internally between agents. 

A pair $(\allocation, \mathbf{p})$ is a \emph{budget-feasible allocation with payments}
if $\mathbf{p}$ is both feasible and budget-balanced.
We study the computational problem \localBudget{} (\localBudgetAbrv{}) of deciding whether there exists a budget-feasible envy-free allocation with payments.

\problemdef{\localBudgetAbrv{}}{decide}{
	A set $G$ of $m$ goods, a set $A$ of $n$ agents, each with a valuation function $v_a \colon \mathcal{P}(G) \to \N$ and a budget $b_a \in \N$.
}{
	Does there exist a budget-feasible allocation with payments $(\allocation,\mathbf{p})$ that is envy-free?
}{}

We adopt standard notations from \emph{parameterized complexity} and algorithm design as introduced by \citet{dblp1894779}, \citet{FELLOWS200953}, \citet{Nie06}, and \citet{doi:10.1137/S0097539703427203}.
A \emph{parameterized problem} is a subset $L \subseteq \Sigma^* \times \mathbb{N}$, where $\Sigma$ is a finite alphabet. The integer $k$ is called the \emph{parameter} of an instance $(x, k)$. The problem $L$ is \emph{fixed-parameter tractable} (FPT) if there exists an algorithm that decides whether $(x, k) \in L$ (a \texttt{yes}-instance) or $(x, k) \notin L$ (a \texttt{no}-instance) in time $f(k) \cdot |x|^{\mathcal{O}(1)}$, where $f$ is a computable function depending only on $k$. The class of all such problems is denoted by FPT. The class XP consists of parameterized problems solvable in time $|x|^{f(k)}$ for some computable function $f$. A problem is \emph{para-NP-hard} with respect to a parameter $k$ if it remains NP-hard even when $k$ is fixed to a constant.

\section{Initial Observations}
\label{sec:observations}

Before tackling the general problem of computing a fair allocation from scratch under individual budget constraints, we begin with the simpler question of whether a given allocation can be made envy-free through payments that respect individual budgets. We show that this variant can be solved in polynomial time using linear programming.

Next, we examine the general case by analyzing how tighter individual budget constraints affect feasibility and fairness. While unconstrained or sufficiently large budgets recover tractable cases from prior work, even minor restrictions can lead to infeasibility or computational hardness.

\subsection{When the Allocation is Given}

We first study \initialfair{} (\initialfairAbrv), the variant in which an allocation is given.

\problemdef{\initialfairAbrv}{init}{
	A set $G$ of $m$ goods, a set $A$ of $n$ agents, each with a valuation function $v_a \colon \mathcal{P}(G) \to \N$ and a budget $b_a \in \N$ as well as an initial allocation $\allocation$.
}{
	Does there exist a payment vector $\mathbf{p}$ such that $(\allocation, \mathbf{p})$ is budget-feasible and envy-free?
}{}

\begin{theorem}\label{thm:lp}
    \initialfairAbrv \ is solvable in polynomial time.
\end{theorem}

\begin{proof}
Given an initial allocation $\mathcal{X}=(X_a)_{a\in A}$, we want to determine whether there exists a payment vector $\mathbf{p} \in \R^n$ satisfying the following constraints:
\begin{align*}
 &v_a(X_a) + p_a \geq v_a(X_{a'}) + p_{a'} && \forall a, a' \in A \\
 &p_a \geq -b_a && \forall a \in A \\
 &\sum_{a\in A} p_a = 0 &&
\end{align*}
Since the number of variables and constraints in the linear program (LP) is polynomial in $n$, and all coefficients are integers, 
feasibility can be decided in polynomial time using standard LP solvers.
Thus, \initialfairAbrv \ is solvable in polynomial time when $\allocation$ is given.
\end{proof}

\subsection{The role of the budget}\label{budget}

We next investigate how the individual budget constraints affect both the existence and the computational complexity of envy-free outcomes. 

We begin by recalling a characterization of envy-freeable allocations established by
\citet{DBLP:conf/sagt/HalpernS19}.

\begin{proposition}[\citet{DBLP:conf/sagt/HalpernS19}]
\label{thm:halpern}
Under positive additive utilities, the following statements are equivalent for an allocation $\allocation$:
\begin{itemize}
    \item[(a)] $\allocation$ is envy-freeable.
    \item[(b)] $\allocation$ maximizes the utilitarian welfare over all permutations of its bundles, i.e., for every permutation $\sigma$ of $A$, it holds that $\sum_{a \in A} v_a(X_a) \geq \sum_{a \in A} v_a(X_{\sigma(a)})$.
\end{itemize}
\end{proposition}

Moreover, they show that an envy-freeable allocation always exists, and thus envy-freeness can always be attained through suitable subsidy payments.
Note that subtracting the average subsidy from each agent’s payment leads to payment transfers that are budget-balanced.
While their model permits unconstrained monetary transfers, our setting imposes agent-specific budget constraints, making it more restrictive. As a result, any necessary condition in their model must also hold in ours, implying the following: 

\begin{observation}
If there exists a payment vector under which an allocation $\allocation$ is budget-feasible and envy-free, then it must be envy-freeable.
\end{observation}

However, envy-freeability alone does not imply budget-feasibility, as demonstrated in the following example.

\begin{example}\label{ex:notfreeable}
Let $n+1$ agents each value a single good $g$ at $n$, and suppose
that each agent has budget $n-1$. Allocate $g$ to an arbitrary agent
$a$. Since all agents have identical valuations, this allocation
maximizes utilitarian welfare among all permutations of the bundles
and is therefore envy-freeable. Furthermore, in any budget-feasible outcome, agent $a$'s utility is at least $v_a(g)+p_a\ge n-b_a=1$.
Envy-freeness therefore requires each of the remaining $n$ agents,
who receive empty bundles, to receive a payment of at least $1$.
This requires a total payment of at least $n$, whereas agent $a$ can
pay at most $n-1$, contradicting budget balance.
\end{example}

The issue is that envy-eliminating payments may exceed agents' budgets, and natural allocation rules can fail.

\begin{example}
Consider three agents and three goods with valuations and budgets as follows:
\[
\begin{array}{c|ccc|c}
 & x_1 & x_2 & x_3 & b_a \\ \hline
1 & 4 & 1 & 1 & 2 \\
2 & 6 & 2 & 2 & 0 \\
3 & 4 & 1 & 1 & 0
\end{array}
\]

The highest-value rule assigns all goods to agent~$2$, but this allocation cannot be made envy-free because agent~$2$ has no budget to compensate agents~$1$ and~$3$. An equal split can also not be made envy-free. In contrast, assigning $x_1$ to agent~$1$ and $\{x_2,x_3\}$ to agent~$2$, together with payments $p=(-2,0,2)$, yields a budget-feasible envy-free outcome.
\end{example}

\begin{observation}\label{notenvy-freeable}
If an allocation $\allocation$ is envy-freeable but the LP of \cref{thm:lp} is infeasible, then no envy-free payment vector can simultaneously satisfy all agents’ budget constraints.
\end{observation}

Since $\allocation$ is envy-freeable, there exists a payment vector that renders all agents envy-free. As discussed above, this payment can be adjusted to form a set of budget-balanced transfer payments. However, since the LP is infeasible, no payment vector can make $(\allocation,\mathbf{p})$ both envy-free and budget-feasible. Thus, every envy-eliminating payment vector requires at least one agent to pay more than their budget allows.

Nevertheless, when budgets are sufficiently large, they no longer restrict feasible allocations, so the models 
align:

\begin{lemma}\label{subsidy}
If each agent's budget satisfies
$b_a \geq \bigl(1-\tfrac{1}{n}\bigr)\max_{g\in G} v_a(g)$,
then there exists a budget-feasible allocation with payments that is envy-free, i.e.,
the instance is a \texttt{yes}-instance of \localBudgetAbrv{}. Moreover, such an allocation and payments can be computed in polynomial time.
\end{lemma}

Let $\mathcal X=(X_a)_{a\in A}$ be an allocation. The \emph{envy graph}
$G_{\mathcal X}$ is the complete directed graph with vertex set $A$.
For every pair of agents $a,a'\in A$, the weight of the arc $(a,a')$ is
\[
	w_{\mathcal X}(a,a') := v_a(X_{a'})-v_a(X_a).
\]
A directed cycle $C=(a_1,a_2,\dots,a_r,a_1)$ in $G_{\mathcal X}$ is called
a \emph{positive envy cycle} if
\[
	w_{\mathcal X}(C)
	:=
	\sum_{i=1}^{r-1} w_{\mathcal X}(a_i,a_{i+1})
	+
	w_{\mathcal X}(a_r,a_1)
	>0.
\]
Equivalently,
\[
	\sum_{i=1}^{r-1}
	\bigl(v_{a_i}(X_{a_{i+1}})-v_{a_i}(X_{a_i})\bigr)
	+
	\bigl(v_{a_r}(X_{a_1})-v_{a_r}(X_{a_r})\bigr)
	>0.
\]

We use the following characterization of envy-freeable allocations due
to~\citet{DBLP:conf/sagt/HalpernS19}. In addition to conditions~(a)
and~(b) stated in \cref{thm:halpern}, they show that the following
condition is equivalent:
\begin{itemize}
    \item[\textit{(c)}] \textit{The envy graph of $\allocation$ contains no positive
    envy cycle.}
\end{itemize}
Moreover, when the envy graph contains no positive envy cycle,
envy-eliminating payments can be obtained as follows.
For every agent $a$, let $q_a$ be the maximum weight of a simple directed path starting at $a$ in the envy graph. 

Then $q_a\geq 0$, and for every pair of agents $a,a'\in A$, $q_a \geq w_{\mathcal X}(a,a')+q_{a'}$ \cite{DBLP:conf/sagt/HalpernS19}. Interpreting $q_a$ as the subsidy received by agent $a$, the payment vector $\mathbf q=(q_a)_{a\in A}$ eliminates envy.

We now recall the allocation procedure of \citet{oneDollar}. Without loss of generality, assume that the number of items is divisible by the number of agents. Indeed, we may add dummy goods of value zero to every agent until the total number of goods is divisible by $n=|A|$. These dummy goods can be discarded from the final allocation, since doing so does not affect any envy weights or $V_a$.

\Cref{alg:bounded-subsidy} then proceeds in rounds: in each round, it computes a maximum-weight \emph{perfect} matching between the agents and the remaining goods, allocates the matched goods, and removes them from further consideration. The procedure terminates once all goods have been allocated.

\begin{algorithm}[t] \caption{\textsc{Bounded-Subsidy Allocation}} \label{alg:bounded-subsidy} \begin{algorithmic}[1] \REQUIRE Agents $A$, goods $G$, additive valuations $(v_a)_{a\in A}$ \ENSURE An allocation $\mathcal X=(X_a)_{a\in A}$ \STATE $X_a \gets \emptyset$ for every $a\in A$ \STATE $t \gets 1$; $G^1 \gets G$ \WHILE{$G^t \neq \emptyset$} \STATE Compute a maximum-weight matching $M^t=\{(a,g_a^t)\mid a\in A\}$ in the complete bipartite graph between $A$ and $G^t$, where edge $(a,g)$ has weight $v_a(g)$. \STATE $X_a \gets X_a\cup\{g_a^t\}$ for every $a\in A$ \STATE $G^{t+1}\gets G^t\setminus \{g_a^t\mid a\in A\}$ \STATE $t\gets t+1$ \ENDWHILE \STATE \textbf{return} $\mathcal X=(X_a)_{a\in A}$ \end{algorithmic} \end{algorithm} The following lemma summarizes the part of the analysis by \citet{oneDollar} that we use in this work.

\begin{lemma}\label{lem:bounded-subsidy-property} Let $\mathcal X=(X_a)_{a\in A}$ be the allocation returned by \Cref{alg:bounded-subsidy}, and let $g_a^t$ denote the good assigned to agent $a$ in round $t$, where $T$ denotes the total number of rounds. For every agent $a\in A$, let $V_a := \max_{g\in G}v_a(g)$. There exists a modified valuation profile $\bar v$ such that the following properties hold: \begin{enumerate} \item the envy graph of $\mathcal X$ with respect to $\bar v$ contains no positive envy cycle; \item if $\bar w(a,a')$ denotes the envy of agent $a$ towards agent $a'$ under $\bar v$ and allocation $\mathcal X$, then for all $a,a'\in A$, \begin{align*} \bar w(a,a') &\geq v_a(g_{a'}^T)-v_a(g_a^1) \geq -V_a \ ; \end{align*}\item for all $a,a'\in A$ we have $\bar w(a,a')\geq w_{\mathcal X}(a,a')$.\end{enumerate} Moreover, $\mathcal X$ and the modified profile $\bar v$ can be constructed in polynomial time. \end{lemma} \begin{proof} The allocation $\mathcal X$ is computed by the repeated maximum-weight matching procedure of \citet{oneDollar}. Their analysis constructs a modified valuation profile $\bar v$, defined as $\bar v_a(g_{a'}^t)=\max(v_a(g_{a'}^t),v_a(g_a^{t+1}))$ for $t<T$, and $\bar v_a(g_{a'}^T)=v_a(g_{a'}^T)$ for every pair of agents $a,a'\in A$. By Claim 4.2. \citep{oneDollar} the allocation $\mathcal{X}$ remains envy-freeable under these valuations, and hence, by the equivalence of conditions~(a)--(c) in
\cref{thm:halpern}, the corresponding envy graph contains no positive
envy cycle. Claim 4.3 \citep{oneDollar} states that the subsidy to each agent is at most $1$ for the modified valuation profile $\bar v$.
As they consider normalized valuations, in our case this implies that, for every pair of agents $a,a'\in A$, \begin{align*} \bar w(a,a') &\geq \sum_{t=1}^{T-1} \bigl(v_a(g_a^{t+1})-v_a(g_a^t)\bigr) + (v_a(g_{a'}^T)-v_a(g_a^T))\\ &= v_a(g_{a'}^T)-v_a(g_a^1) \geq -v_a(g_a^1) \geq -V_a. \end{align*} Finally, the construction of $\bar v$ leaves every agent's value for her own assigned goods unchanged and can only increase her value for goods assigned to other agents. Hence the modified envy of any agent towards any other agent is at least the original envy. This gives $\bar w(a,a')\geq w_{\mathcal X}(a,a')$ for all $a,a'\in A$. The allocation and the modified valuation profile can be constructed in polynomial time. \end{proof}

We can now prove \cref{subsidy}. \begin{proof}[Proof of \cref{subsidy}] For every agent $a\in A$, let $V_a := \max_{g\in G} v_a(g)$. Apply \Cref{alg:bounded-subsidy}, and let $\mathcal X=(X_a)_{a\in A}$ be the resulting allocation. Let $\bar v$ be the modified valuation profile from \cref{lem:bounded-subsidy-property}, and let $\bar w(a,a')$ be the envy of agent $a$ towards agent $a'$ under $\bar v$ and allocation $\mathcal X$. By \cref{lem:bounded-subsidy-property}.1, the envy graph of $\mathcal X$ under $\bar v$ contains no positive envy cycle. For every agent $a\in A$, let $q_a$ be the maximum weight of a directed path starting at $a$ in this envy graph. Since the empty path has weight zero, we have $q_a\geq 0$. Moreover, for every pair of agents $a,a'\in A$, $q_a \geq \bar w(a,a')+q_{a'}$.

Therefore, $q_{a'}-q_a \leq -\bar w(a,a')$. By \cref{lem:bounded-subsidy-property}.2, $\bar w(a,a')\geq -V_a$. Hence $q_{a'}-q_a \leq V_a$ for all $a,a'\in A$. 

By \cref{lem:bounded-subsidy-property}.3, we have $\bar w(a,a')\geq w_{\mathcal X}(a,a')$ for all $a,a'\in A$. Equivalently, \[ v_a(X_{a'})-v_a(X_a) \leq \bar v_a(X_{a'})-\bar v_a(X_a). \] Interpreting $q_a$ as the subsidy received by agent $a$, a payment vector $\mathbf q=(q_a)_{a\in A}$ that eliminates envy under $\bar v$ also eliminates envy under the original valuation profile $v$. Therefore, $\mathbf q$ eliminates envy for the original instance as well. It remains to transform the non-negative subsidies $\mathbf q$ into budget-balanced payments $\mathbf p$. Let $\bar q := \frac{1}{n}\sum_{a\in A} q_a$ and define $p_a := q_a-\bar q $ for every agent $a\in A$. Then the payments are budget-balanced: $\sum_{a\in A}p_a = \sum_{a\in A}q_a-n\bar q = 0$. Furthermore, envy-freeness is preserved because all payment differences remain unchanged. Indeed, for all agents $a,a'\in A$, \[p_a-p_{a'} = (q_a-\bar q)-(q_{a'}-\bar q) = q_a-q_{a'}.\] Thus $(\mathcal X,\mathbf p)$ is envy-free under the original valuation profile. Finally, we verify the individual budget constraints. For every agent $a\in A$,
\[
-p_a = \bar q-q_a = \frac{1}{n}\sum_{a'\in A}(q_{a'}-q_a)=
\frac{1}{n}\sum_{a'\in A\setminus\{a\}}(q_{a'}-q_a),
\]
where the last equality holds because the summand for $a'=a$ is $0$.
Using $q_{a'}-q_a\leq V_a$ for every $a'\in A$, we obtain
\[
-p_a \leq \frac{1}{n}\sum_{a'\in A\setminus\{a\}}V_a = \frac{n-1}{n}\,V_a.
\]
By assumption, $b_a\geq (1-\frac{1}{n})V_a$. Hence $-p_a\leq b_a$ for every agent $a\in A$.
Therefore the payment vector $\mathbf p$ is budget-feasible: it is budget-balanced and no agent pays more than her budget. All steps can be carried out in polynomial time. The allocation is computed by polynomially many maximum-weight matchings. The subsidy payments $\mathbf q$ are obtained from maximum path weights in an envy graph with no positive envy cycle \citep{DBLP:conf/sagt/HalpernS19}, which can be computed in polynomial time. The final payments $\mathbf p$ are obtained as described above. Hence, under the assumption $b_a\geq (1-\frac1n)V_a$ for every agent $a \in A$, the desired envy-free allocation with budget-feasible payments can be computed in polynomial time.
\end{proof}

Lemma~\ref{subsidy} identifies a sufficient budget threshold under which envy-freeness is always achievable and efficiently computable. The following theorem shows that this threshold is essentially tight.

\begin{theorem}\label{thm:np}
    \localBudgetAbrv\ is NP-complete, even when all agents have budget $b_a = \max_{g \in G} v_a(g) - 1$.
\end{theorem}

\begin{proof}
We give the full reduction from \textsc{Unary Exact Bin Packing} \citep{DBLP:journals/corr/abs-2009-11840}.
Let $(s_1,\dots,s_m;B,k)$ be an instance with
$\sum_{i=1}^m s_i=kB$.
We may assume that $B\ge 2$ and $\max_i s_i=B$: otherwise, we add one
item of size $B$ and one additional bin.

We construct an instance
$I'=(A,G,v,\mathbf b)$ of \localBudgetAbrv{} as follows.
There are $(k+1)B-k$ type-$a$ agents
$a_1,\ldots,a_{(k+1)B-k}$ and $k$ type-$b$ agents
$b_1,\ldots,b_k$.
For every item $s_i$, we introduce a good $g_i$ with
$v_{a_r}(g_i)=0$
and
$v_{b_j}(g_i)=s_i$
for all relevant $r,j$. We also introduce a special good $g^\star$
with
$v_{a_r}(g^\star)=B$
and
$v_{b_j}(g^\star)=0$.
Finally, every agent has budget $B-1$.

Since $\max_i s_i=B$, every agent's maximum single-good value is $B$,
and hence $b_a=\max_{g\in G}v_a(g)-1$ for every agent $a$. Moreover, the construction is polynomial because the source instance is encoded in unary.

\noindent\textbf{Correctness.}
We show that $I$ is a \texttt{yes}-instance of
\textsc{Unary Exact Bin Packing} if and only if $I'$ admits a budget-feasible,
envy-free allocation with balanced payments.

\smallskip
\noindent($\Rightarrow$)
Assume $S$ admits a packing into $k$ bins of capacity $B$, i.e., there is a
partition $S_1,\dots,S_k$ of $S$ such that
$\sum_{s\in S_j}s=B$ for all $j\in[k]$.

Define an allocation $\mathcal{X}$ by giving each $b_j$ exactly the goods
corresponding to $S_j$,
$X_{b_j}=\{g_i : s_i\in S_j\}$
for all $j\in[k]$,
giving one arbitrary type-$a$ agent, say $a_1$, the special good, $X_{a_1}=\{g^\star\}$,
and setting $X_{a_r}=\emptyset$ for all other type-$a$ agents $a_r\neq a_1$.
Define payments $\mathbf{p}$ by
\begin{align*}
&p_{b_j}=-B+1 \quad\text{for all } j\in[k],\\
&p_{a_1}=-B+1,\\
&p_{a_r}=1 \quad\text{for all } a_r\neq a_1.
\end{align*}
A type-$a$ agent evaluates the special-good recipient at $B-(B-1)=1$, every other type-$a$ agent at $1$, and every type-$b$ bundle at $0-(B-1)\le 1$.
A type-$b$ agent evaluates every type-$b$ bundle at $B-(B-1)=1$, every (nonspecial) type-$a$ agent at $0+1=1$, and the special-good recipient at $0-(B-1)\le 1$. Hence the outcome
is envy-free. Moreover, every paying agent pays exactly $B-1$, so $-p_a\le B-1=b_a$ for every agent $a$. Finally, since there are
$(k+1)(B-1)$
type-$a$ agents other than $a_1$, we have
\[\sum_{a\in A}p_a
=(k+1)(B-1)-(B-1)-k(B-1)=0.\]
Hence, the payments are budget-balanced. It follows that the allocation with payments is budget-feasible.

\smallskip
\noindent($\Leftarrow$)
Assume there exists a budget-feasible, envy-free outcome
$(\mathcal{X},\mathbf{p})$.

First, we claim that the special good $g^\star$ must be allocated to a type-$a$ agent.
Suppose instead that $g^\star$ is allocated to a type-$b$ agent $b_j$. Then
every type-$a$ agent values $X_{b_j}$ at $B$. Since $b_j$ can pay at most
$B-1$, envy-freeness implies that every type-$a$ agent must receive payment at
least $1$. Thus the total positive payment to type-$a$ agents is at least
$(k+1)B-k$.
However, the only remaining agents are the $k$ type-$b$ agents, and together
they can pay at most $k(B-1)$, which is strictly smaller. This contradicts
balancedness of the payments. Hence $g^\star$ is allocated to some type-$a$ agent; call this
agent $a^\star$.

Since $a^\star$ values $g^\star$ at $B$ and can pay at most $B-1$, we have
$v_{a^\star}(X_{a^\star})+p_{a^\star}\ge B-(B-1)=1$.
Every type-$a$ agent $a_r\neq a^\star$ values her own bundle at $0$, but values $X_{a^\star}$ at $B$. Hence, by envy-freeness, every such agent must receive payment at least $1$. There are
\[
(k+1)B-k-1=(k+1)(B-1)
\]
such agents. Thus their total payment is at least $(k+1)(B-1)$.
The remaining agents are $a^\star$ and the $k$ type-$b$ agents. Each of them can
pay at most $B-1$, so together they can pay at most $(k+1)(B-1)$. Since payments
are balanced, all these bounds must be tight. In particular,
$p_{b_j}=-(B-1)$
 for all $j\in[k]$.

We next show that all goods $g_i$ are allocated to type-$b$ agents. Suppose, for
contradiction, that some $g_i$ is allocated to a type-$a$ agent. Then the total
value of the goods allocated to type-$b$ agents is strictly less than $kB$, so
some type-$b$ agent $b_j$ receives a bundle of value strictly less than $B$.
Since $p_{b_j}=-(B-1)$, we get
\[
v_{b_j}(X_{b_j})+p_{b_j}<B-(B-1)=1.
\]
But every type-$a$ agent other than $a^\star$ receives payment at least $1$, so
$b_j$ would envy such an agent, a contradiction. Hence all goods $g_i$ are
allocated among the type-$b$ agents.

For each $j\in[k]$, define 
$S_j := \{\, s_i \in S : g_i \in X_{b_j}\,\}$. 
Then $(S_1,\dots,S_k)$ is a partition of $S$. It remains to show that each part
has total size $B$.

Suppose not. Then there are $j\neq j'$ such that
$\sum_{s\in S_j}s > \sum_{s\in S_{j'}}s$.
Since type-$b$ agents have identical valuations over the goods $g_i$, we have $v_{b_{j'}}(X_{b_j}) > v_{b_{j'}}(X_{b_{j'}})$.
Moreover, $p_{b_j}=p_{b_{j'}}=-(B-1)$. Hence
\[
v_{b_{j'}}(X_{b_j})+p_{b_j}
>
v_{b_{j'}}(X_{b_{j'}})+p_{b_{j'}},
\]
so $b_{j'}$ envies $b_j$, contradicting envy-freeness.

Therefore all sums $\sum_{s\in S_j}s$ are equal. Since their total is $kB$, each of them equals $B$. Thus $(S_1,\dots,S_k)$ is an exact bin packing of $S$.

Hence $S$ is a \texttt{yes}-instance of \textsc{Unary Exact Bin Packing} if and only if $I'$ has a budget-feasible, envy-free allocation with payments.

Membership in NP holds because a proposed allocation can be verified for envy-freeness in polynomial time. 
\end{proof}

We next show that hardness persists even under different strong restrictions on the valuation functions.

\begin{theorem}\label{thm:npw}
\localBudgetAbrv{} is NP-complete, even when all agents have binary
valuations and identical budgets. It is NP-complete already for $n=2$
agents with identical, binary-encoded additive valuations.
Furthermore, parameterized by the number $n$ of agents,
\localBudgetAbrv{} is $W[1]$-hard, even when all agents have identical
additive valuations encoded in unary.
\end{theorem}
\begin{proof}
    Consider the case where all budgets are zero, i.e., $b_a = 0$ for all $a \in A$. In this setting, \localBudgetAbrv\  reduces to the problem of deciding whether an envy-free allocation exists. This problem is known to be NP-hard even when agents have binary valuations, as shown in Theorem~11 of \citet{AZIZ201571}. Membership in NP holds because a proposed allocation can be verified for envy-freeness in polynomial time.
    Thus, the problem is NP-complete already under binary valuations and identical budgets. 
For two agents with identical additive valuations, an allocation is
envy-free if and only if both agents receive the same value. Thus,
\textsc{Partition} yields NP-hardness already for $n=2$ when valuations
are encoded in binary. Finally, deciding the existence of an envy-free allocation is $W[1]$-hard with respect to the number of agents $n$, as follows from reductions from Unary Bin Packing~\cite{BBN16}. 
\end{proof}

\begin{remark}\label{rem:n-equals-m}
\localBudgetAbrv{} remains NP-complete when the numbers of agents and goods coincide, i.e., $n=m$ 
\end{remark}

\begin{proof}
We adapt the NP-hardness reduction for envy-free allocation with identical additive valuations and zero budgets. Let $I$ be an instance with $n$ agents and $m$ goods, where all agents have the same non-negative additive valuation $v$ and $v(G)=nB$. We may assume that $m\ge n$.

Construct an instance $I'$ by adding a new good $g^\star$, one new agent with the same valuation as the original agents, and $m-n$ dummy agents. Extend the valuation of every non-dummy agent by setting $v(g^\star)=B$, while every dummy agent values all goods at zero. All agents have budget zero. The resulting instance has $n+1+(m-n)=m+1$ agents and $m+1$ goods, and hence equally many agents and goods.

$(\Rightarrow)$ Suppose first that $I$ admits an envy-free allocation. Since the
original agents have identical valuations and the total value is
$nB$, every agent receives value $B$. Assigning $g^\star$ to the new
non-dummy agent and leaving all dummy agents empty therefore yields an
envy-free allocation of $I'$.

$(\Leftarrow)$ Conversely, suppose that $I'$ admits an envy-free allocation. Since all budgets are zero, all payments must be zero. The good $g^\star$ cannot be assigned to a dummy agent: otherwise, every non-dummy agent would need utility at least $B$ to avoid envying its bundle, although the
original goods have total value only $nB$ for the $n+1$ non-dummy
agents.

Thus, $g^\star$ is assigned to a non-dummy agent. The $n+1$ non-dummy
agents have identical valuations, so envy-freeness requires them to
receive the same value. Their total value for all goods is $v(G)+v(g^\star)=(n+1)B$, and hence each receives value exactly $B$. In particular, the
recipient of $g^\star$ receives no original goods of positive value. If the recipient of $g^\star$ is one of the original agents, we relabel the new non-dummy agent as that original agent. Since all non-dummy agents have identical valuations, this does not affect envy-freeness. 
Remove this agent and $g^\star$, and reassign any zero-valued goods in
her bundle arbitrarily among the remaining agents. The resulting
allocation of the original goods gives each of the $n$ original agents
value $B$ and is therefore envy-free.

Consequently, $I$ admits an envy-free allocation if and only if $I'$
does. The construction is polynomial and satisfies $n=m$, proving the
claim.
\end{proof}

\section{Fixed-parameter Algorithms}
\label{sec:algorithms}
\appendixsubsection{sec:algorithms}

While the previous section established the NP-completeness of \localBudgetAbrv{}, we now refine our understanding by analyzing the problem through the lens of parameterized complexity. We investigate tractable cases and structural insights that enable fixed-parameter algorithms under certain assumptions. We begin with the general case of heterogeneous valuations, then focus on the special case of identical valuations, and finally explore structured instances defined by agent and good \emph{types}.

\subsection{Differing valuations}

The presence of payments allows agents to receive positive utility even without being allocated any goods. As a result, allocations are no longer strictly bounded by a function of the number of available goods.

Nevertheless, we now present structural observations that remain valid even when valuations differ across agents. These will be useful for analyzing the parameterized complexity of \localBudgetAbrv{} with respect to the number of goods $m$.

\begin{lemma}\label{obs:m}
Let $(\allocation,\mathbf{p})$ be an envy-free, budget-feasible allocation, and let $X\subseteq G$ be a non-empty bundle allocated to an agent~$a$. Then at most $m-1$ agents value $X$ strictly higher than $a$ does.
\end{lemma}

\begin{proof}

Suppose, for a contradiction, that there is a set $A'\subseteq A\setminus\{a\}$ of at least $m$ agents such that $v_{a'}(X)>v_a(X)$ for every $a'\in A'$.
Since an allocation contains at most $m$ non-empty bundles and $a$ receives one of them, some agent $a'\in A'$ must receive an empty bundle. Envy-freeness implies
$v_{a'}(X)+p_a\le p_{a'}\le v_a(X)+p_a$.
Since $v_{a'}(X)>v_a(X)$, these inequalities are incompatible, a contradiction. 
\end{proof}

\begin{lemma}\label{obs:b}
Let $(\allocation,\mathbf{p})$ be an envy-free, budget-feasible allocation in which agent~$a$ receives a bundle $X$ and agent~$a'$ receives an empty bundle. If $v_a(X)=v_{a'}(X)$ and $b_{a'}\ge b_a$, then swapping the bundles and payments of $a$ and $a'$ preserves envy-freeness and budget-feasibility.
\end{lemma}

\begin{proof}
By envy-freeness between $a$ and $a'$, we have $v_a(X)+p_a\ge p_{a'}$ and  $p_{a'}\ge v_{a'}(X)+p_a$. Since $v_a(X)=v_{a'}(X)$, both inequalities are tight, and thus $p_{a'}=v_a(X)+p_a$. Now give $X$ and payment $p_a$ to $a'$, and give the empty bundle and payment $p_{a'}$ to $a$. The bundle--payment pairs available for comparison are merely exchanged, and the equality above ensures that the utilities of $a$ and $a'$ remain unchanged. Hence envy-freeness is preserved.
For budget-feasibility, agent~$a$ receives payment $p_{a'}\ge p_a\ge -b_a$. Moreover, as $b_{a'}\ge b_a$, $p_a\ge-b_a\ge-b_{a'}$. Thus neither agent exceeds her budget, while budget balance remains.
\end{proof}

These observations bound the set of agents who need to be considered as recipients of any fixed bundle.

\begin{theorem}\label{thm:m} \localBudgetAbrv\ can be solved in $\mathrm{poly}(n,m) \cdot m^{2m}$ time, which is fixed-parameter tractable with respect to the number of goods $m$.\end{theorem} 

\begin{proof} 
We may assume that $n \geq m$, since otherwise all agents can
be retained as candidates. For every non-empty bundle $X\subseteq G$, let $C_X$ consist of the first $m$ agents ordered by non-increasing value $v_a(X)$, breaking ties by non-increasing budget. We first show that, whenever a feasible solution exists, there is one in which every non-empty bundle $X$ is assigned to an agent in $C_X$. Consider a bundle $X$ assigned to some agent $a\notin C_X$. By \cref{obs:m}, fewer than $m$ agents value $X$ strictly more than $a$ does. Hence some agent $a'\in C_X$ satisfies $v_{a'}(X)=v_a(X)$ and $b_{a'}\ge b_a$. Moreover, since $C_X$ contains $m$ agents whereas there are at most $m-1$ other non-empty bundles, one such agent $a'\in C_X$ receives the empty bundle. By \cref{obs:b}, we may swap $a$ and $a'$ while preserving envy-freeness and budget-feasibility. Thus $X$ can be assigned to one of its candidates. The algorithm first enumerates all partitions of the $m$ goods into at most $m$ labeled bundles. There are at most $m^m$ such partitions. For every non-empty bundle $X$, it then chooses a distinct recipient from $C_X$. Since $|C_X|\le m$, there are at most another $m^m$ choices. For each resulting allocation, we use the linear program from \cref{thm:lp} to test whether suitable budget-feasible, envy-eliminating payments exist. Each LP is solvable in polynomial time, and the total number of allocations considered is at most $m^{2m}$. The overall running time is therefore $\mathrm{poly}(n,m)\cdot m^{2m}$, which is fixed-parameter tractable with respect to~$m$.
\end{proof}

As the total number of goods $m$ can be extremely large, especially when goods of the same type are duplicated and good valuations can be high, we focus on a more refined parameterization. Specifically, we consider the number of agents $n$ and the number of distinct good types $\lambda$. For this setting, we establish FPT using a mixed integer linear programming (MILP) formulation.

\begin{proposition}[{\citet{lenstra1983integer,kannan1987minkowski,frank1987application}}]
\label{prop:lenstra}
There is an algorithm for solving
MILP with $p$ integer variables
in time $2^{O(p \log p)} \cdot \mathrm{poly}(L)$ and space polynomial in $L$,
where $L$ is the bit-size of the input.
\end{proposition}

\begin{theorem}\label{thm:fptnt}
\localBudgetAbrv \ is fixed-parameter tractable with respect to the number of agents and 
good types.
\end{theorem}

\begin{proof}
The MILP uses an integer variable $0 \leq x_{i,\ell} \leq m_\ell$ for each agent $a_i \in A$ and each good type $\ell \in [\lambda]$, where $x_{i,\ell}$ denotes the number of goods of type $\ell$ allocated to agent $a_i$, and $m_\ell$ is the total available supply of good type $\ell$. Additionally, we use variables $p_{a_i} \in \R$ to denote the net payments of each agent $a_i \in A$. With slight abuse of notation, we refer to $g_\ell$ as a good of type $\ell$. The MILP includes the following constraints:
\begin{align*}
\sum_{\ell \in [\lambda]}x_{i,\ell} \cdot v_{a_i}(g_\ell) + p_{a_i} &\geq \sum_{\ell \in [\lambda]}x_{j,\ell} \cdot v_{a_i}(g_\ell) + p_{a_j} \\ &\qquad \qquad \forall a_i, a_j \in A \\ 
p_{a_i}&\geq -b_{a_i} \qquad \forall a_i\in A\\
\sum_{a_i\in A}x_{i,\ell}&=m_\ell \quad \qquad \forall \ell\in [\lambda]\\
\sum_{a_i\in A}p_{a_i}&=0 &
\end{align*}

The first constraint ensures that the allocation with payments is envy-free, while the second ensures that no agent exceeds their budget. The third constraint guarantees allocation completeness by requiring that all goods of each type $\ell$ are fully distributed, and the fourth enforces budget balance of the payments.

Our MILP contains $n\cdot\lambda$ integer variables (one per agent--type pair) and $\mathcal{O}(n^2)$ constraints.
Fixed-parameter tractability now follows from Proposition~\ref{prop:lenstra}.
\end{proof}

Let $v_{\max}$ be the maximum good valuation. Since the number of good types is bounded by $(v_{\max} + 1)^n$, we obtain:

\begin{corollary}\label{thm:value} \localBudgetAbrv{} is fixed-parameter tractable with respect to the number of agents $n$ and the maximum valuation of any good. \end{corollary}

\subsection{When agents have identical valuations}

Next we look at agents with identical valuations. This case captures symmetric situations where agents value all goods equally, for example identical shifts or machines.

Let $v$ denote the common valuation function shared by all agents. Under this assumption, the requirement of envy-freeness imposes significant structural constraints on 
any 
allocation.

\begin{lemma}\label{thm:identical}
Let $(\mathcal{X},p)$ be an envy-free allocation with budget-balanced payments among $n$ agents with identical valuations $v$. Then every agent receives the same effective utility $u_a(\allocation,\mathbf{p}) =\frac{v(G)}{n}$.

\end{lemma}

\begin{proof}  

For any two agents $a,a' \in A$, envy-freeness requires
$v_a(X_a) + p_a \;\geq\; v_a(X_{a'}) + p_{a'}$, 
and 
$v_{a'}(X_a) + p_a \;\leq\; v_{a'}(X_{a'}) + p_{a'}$,
which, as $v = v_a = v_{a'}$ under identical valuations, implies
$v(X_a) + p_a = v(X_{a'}) + p_{a'} $ for all $a,a' \in A$.
Hence, all agents must receive the same effective utility in any envy-free allocation with payments.  

Additionally, as payments are budget-balanced and allocations complete, the total effective utility across all agents is
\[
\sum_{a \in A} \big(v(X_a) + p_a\big)
= \sum_{a \in A} v(X_a) + \sum_{a \in A} p_a
= v(G) + 0.
\]
Therefore, for all agents to receive the same utility, we have
$u_a(\allocation,\mathbf{p}) = v(X_a) + p_a = \frac{v(G)}{n}$ for all $a \in A$, in every envy-free, budget-feasible allocation with payments $(\allocation,\mathbf{p})$.
\end{proof}

This structural characterization significantly constrains the space of feasible allocations. We use this to show that \localBudgetAbrv{} can be solved more efficiently under identical valuations.

\begin{theorem}\label{thm:fptm}
For agents with identical valuations, \localBudgetAbrv{} can be solved in 

$\mathrm{poly}(n,m)\cdot 2^{m}$ time.
\end{theorem}

\begin{proof}
We construct a dynamic programming (DP) table $T$ with entries $T[i, X] \in \{1, 0\}$,
where $i \in [n]$ and $X \subseteq G$. We interpret $T[i, X] = 1$ to mean that the subset $X$ can be 
allocated to the first $i$ agents within their respective budget constraints.
For all $X \subseteq G$,
\[
T[1, X] = 1 \quad \text{if and only if} \quad \frac{v(G)}{n} - v(X) \geq - b_{a_1}.
\]
\[
\text{For $i > 1$, }
T[i, X] = \bigvee_{\substack{X' \subseteq X \\ \frac{v(G)}{n} - v(X') \geq - b_{a_i}}} T[i-1, X \setminus X'].
\]

We prove the correctness of the recurrence in \cref{thm:fptm} by induction on the index $i$.

\textbf{Base Case ($i = 1$):}  
According to \cref{thm:identical}, under identical valuations, each agent's effective utility must equal $\frac{v(G)}{n}$. For each $X \subseteq G$, assigning $X$ to agent $a_1$ thus requires a transfer
\[
p_{a_1} = \frac{v(G)}{n} - v(X).
\]
The allocation is budget-feasible if and only if $p_{a_1} \geq - b_{a_1}$, i.e.,
\[ \frac{v(G)}{n} - v(X) \geq - b_{a_1}.
\]
Thus, the definition of $T[1, X]$ correctly reflects feasibility for $a_1$.

\textbf{Induction Hypothesis:}  
Assume that for some $i \geq 1$, $T[i, X] = 1$ if and only if $X$ can be feasibly allocated to agents $a_1, \dots, a_i$ within their budget constraints.

\textbf{Induction Step:}  
We now show that $T[i+1, X] = 1$ if and only if $X$ can be feasibly allocated to agents $a_1, \dots, a_{i+1}$ under the same condition.

Fix an arbitrary subset $X \subseteq G$. Define $T \coloneq T[i+1, X]$, and let $O \coloneq O(i+1, X)$ indicate whether the set $X$ can be allocated to the first ${i+1}$ agents.

We distinguish between the two cases $O = 1$ and $O = 0$.

\emph{Case $O = 1$:} Then the set $X$ can be split into $X' \subseteq X$, allocated to agent $a_{i+1}$, such that 
\[
  \frac{v(G)}{n} - v(X')   \geq -b_{a_{i+1}},
\]
and the remainder $X \setminus X'$ allocated to agents $a_1, \dots, a_i$. By the induction hypothesis, this implies that
$T[i, X \setminus X'] = 1$.
From the recurrence definition, it follows that $T[i+1, X] = 1$, and hence $O = 1$ implies $T = 1$.

\emph{Case $O = 0$:} Suppose, for contradiction, that $T = 1$. Then $T[i+1, X] = 1$ implies the existence of $X' \subseteq X$ such that:
\[
  \frac{v(G)}{n} - v(X')   \geq - b_{a_{i+1}}
\quad \text{and} \quad
T[i, X \setminus X'] = 1.
\]

By the induction hypothesis, $X \setminus X'$ can be feasibly allocated to $a_1, \dots, a_i$, and $X'$ is feasible for $a_{i+1}$, contradicting $O = 0$. Hence, $T = O$.

Since we have shown that $O = 1$ implies $T = 1$ and $T = 1$ implies $O = 1$ for all $i \in \{1,\dots,n\}$ and all $X \subseteq G$, it follows that $T = O$. Therefore, the dynamic program correctly decides feasibility.

By \cref{thm:identical}, such an allocation corresponds to a budget-feasible envy-freeable allocation, ensuring equal effective utility for all agents under budget-feasibility. Thus, the instance is a \texttt{yes}-instance if and only if
$T[n, G] = 1$.

\textbf{Running time:}
The dynamic programming table $T$ has $n \cdot 2^m$ entries, where $m = |G|$ is the number of goods. For each entry $T[i, X]$, the recurrence examines all subsets $X' \subseteq X$, of which there are at most $2^m$. By applying subset convolution \cite{dblp1894779}, the recurrence can be evaluated efficiently in $\mathcal{O}(n \cdot \text{poly}(m) \cdot 2^m)$ time.

This shows that \localBudgetAbrv{} is fixed-parameter tractable with respect to the number of goods $m$ when agents have identical additive valuations.
\end{proof}

We show that the dynamic program from \cref{thm:fptm} can be
implemented in $\mathrm{poly}(n,m)\cdot 2^m$ time using fast
subset convolution \cite{dblp1894779}.

\begin{proof}
For each agent $a_i$, define the indicator function
\[
F_i(Y)=
\begin{cases}
1, & \text{if } \dfrac{v(G)}{n}-v(Y)\geq -b_{a_i},\\
0, & \text{otherwise},
\end{cases}
\qquad Y\subseteq G.
\]

Let $\widetilde{T}[i,X]$ denote the number of ways to partition
$X$ into feasible bundles for agents $a_1,\ldots,a_i$. We initialize
\[
\widetilde{T}[0,\emptyset]=1
\text{ and }
\widetilde{T}[0,X]=0
\text{ for every nonempty }X\subseteq G.
\]
For every $i\in[n]$ and $X\subseteq G$, we have
\[
\widetilde{T}[i,X]
=
\sum_{Y\subseteq X}
\widetilde{T}[i-1,X\setminus Y]F_i(Y).
\]
For a fixed $Y$, the corresponding term counts all feasible allocations in which agent $a_i$ receives $Y$. Consequently,
\[
T[i,X]=1
\quad\Longleftrightarrow\quad
\widetilde{T}[i,X]>0.
\]
In particular, the instance is a \texttt{yes}-instance if and only if
$\widetilde{T}[n,G]>0$.

For functions $f,h\colon 2^G\to\mathbb{Z}$, their subset convolution
is defined by
\[
(f\star h)(X)
=
\sum_{Y\subseteq X}f(X\setminus Y)h(Y).
\]
Thus, for every $i\in[n]$,
$\widetilde{T}[i,\cdot]
=
\widetilde{T}[i-1,\cdot]\star F_i$.
Using fast subset convolution, all values
$\widetilde{T}[i,X]$, $X\subseteq G$, can be computed from
$\widetilde{T}[i-1,\cdot]$ and $F_i$ in
$\mathrm{poly}(m)\cdot 2^m$ arithmetic operations
\cite{dblp1894779}.

Thus, each agent-stage update \(\widetilde{T}[i-1]\mapsto \widetilde{T}[i]\) costs \(\mathcal{O}(\mathrm{poly}(m)\cdot 2^m)\), and over all \(n\) agents the total running time is $\mathrm{poly}(n,m)\cdot 2^m$.
\end{proof}

While this approach involves enumerating all subsets of $G$, it becomes computationally intensive when the number of distinct goods is large.
However, in many practical settings, the number of distinct good types is small, what varies is only the multiplicity of each type. This structural constraint enables a more efficient solution, as captured in the following lemma.

\begin{lemma}\label{thm:types}

For agents with identical valuations and $\lambda$ good types, \localBudgetAbrv{} can be solved in 
$\mathrm{poly}(n,m)\cdot (1+\frac{m}{\lambda})^{2\lambda}$.
Hence, the problem is polynomial-time solvable for every fixed $\lambda$.
\end{lemma}

\begin{proof}
Let $G$ contain $\ell_j$ goods of type $j \in [\lambda]$. Since all goods of a given type are indistinguishable under identical valuations, we can represent any subset $X \subseteq G$ as a vector $\mathbf{k} = (k_1, \dots, k_\lambda) \in \N^\lambda$, where $k_j$ denotes the number of goods of type $j$ in $X$. 

We define a dynamic programming table, analogous to the one in the proof of \cref{thm:fptm},
$T[i, k_1, \dots, k_\lambda] \in \{0, 1\}$
to indicate whether it is feasible to allocate the bundle represented by $\mathbf{k}$ to the first $i$ agents such that each has sufficient budget to ensure an effective utility of exactly $\frac{v(G)}{n}$.

Here, $i \in [n]$, and for each $j \in [\lambda]$, $k_j \in \{0, 1, \dots, \ell_j\}$. The table is initialized as follows:

\[
T[1, k_1, \dots, k_\lambda] =
\begin{cases}
1 & \text{if } \frac{v(G)}{n} - \sum_{j=1}^{\lambda} k_j \cdot v_j \geq - b_{a_1}, \\
0 & \text{otherwise},
\end{cases}
\]
where $v_j$ denotes the value of a single good of type $j$.

For $i > 1$, define:
\[
T[i, \mathbf{k}] = \bigvee_{\substack{\mathbf{k'} \leq \mathbf{k} \\ \frac{v(G)}{n} - \sum_{j=1}^{\lambda} k_j' \cdot v_j \geq - b_{a_i}}} T[i-1, \mathbf{k} - \mathbf{k'}],
\]
where $\mathbf{k'}, \mathbf{k} \in \N^\lambda$ and $\mathbf{k'} \leq \mathbf{k}$ component-wise.

This ensures that agent $a_i$ receives a feasible bundle $\mathbf{k'}$, and that the remaining set $\mathbf{k} - \mathbf{k'}$ can be feasibly allocated to agents $a_1, \dots, a_{i-1}$.

We proceed by induction on $i$ to prove the correctness of the recurrence.

\textbf{Base Case ($i = 1$):}  
For any $\mathbf{k}$ with $k_j \leq \ell_j$ for all $j \in [\lambda]$, the transfer needed for agent $a_1$ to receive the bundle is:
\[
p_{a_1} = \frac{v(G)}{n} - \sum_{j=1}^{\lambda} k_j \cdot v_j.
\]
We set $T[1, \mathbf{k}] = 1$ if and only if $p_{a_1} \geq - b_{a_1}$, thus correctly reflecting the budget-feasibility for agent $a_1$.

\textbf{Induction Hypothesis:}  
Assume that for all vectors $\mathbf{k}$, $T[i, \mathbf{k}] = 1$ if and only if the bundle $\mathbf{k}$ can be feasibly allocated to agents $a_1, \dots, a_{i}$.

\textbf{Induction Step:} We aim to show that $T[i+1, \mathbf{k}] = 1$ if and only if $\mathbf{k}$ can be feasibly allocated to agents $a_1, \dots, a_{i+1}$ under the same condition.

Fix $\mathbf{k}$, and let $T \coloneqq T[i+1, \mathbf{k}]$ and $O \coloneqq O(i+1, \mathbf{k})$
 indicate whether the set $\mathbf{k}$ can be allocated to the first $i+1$ agents. 

We distinguish between the two cases $O = 1$ and $O = 0$.

\emph{Case $O = 1$:} Then there exists a feasible decomposition $\mathbf{k} = \mathbf{k'} + \mathbf{k''}$ where $\mathbf{k'}$ is allocated to agent $a_{i+1}$ such that:
\[
  \frac{v(G)}{n} - \sum_{j=1}^{\lambda} k_j' \cdot v_j    \geq - b_{a_{i+1}},
\]
and $\mathbf{k''} = \mathbf{k} - \mathbf{k'}$ is feasible for the previous $i$ agents. By the induction hypothesis, $T[i, \mathbf{k''}] = 1$, and the recurrence thus sets $T[i+1, \mathbf{k}] = 1$. Hence, $O = 1$ implies $T = 1$.

\emph{Case $O = 0$:}
Suppose, for contradiction, that $T = 1$. Then there exists some $\mathbf{k'} \leq \mathbf{k}$ such that:
\[
  \frac{v(G)}{n} - \sum_{j=1}^{\lambda} k_j' \cdot v_j    \geq - b_{a_{i+1}}
\quad \text{and} \quad
T[i, \mathbf{k} - \mathbf{k'}] = 1.
\]
By the induction hypothesis, $\mathbf{k} - \mathbf{k'}$ is feasible for $a_1, \dots, a_i$, and agent $a_{i+1}$ has a high enough budget to receive $\mathbf{k'}$. This contradicts $O = 0$, so $T = 0$.

Hence, $T = O$ for all entries.

By \cref{thm:identical}, such an allocation ensures equal effective utility for all agents can be achieved. Therefore, the instance is a \texttt{yes}-instance if and only if 
$T[n, \ell_1, \dots, \ell_\lambda] = 1$.

\textbf{Running time:} The total number of goods is $m = \sum_{j=1}^{\lambda} \ell_j$. 
Applying the AM--GM inequality to the terms $(1+\ell_1),\dots,(1+\ell_\lambda)$, we obtain the bound
$\!\big(1+\tfrac{m}{\lambda}\big)^{\lambda}$.
Thus the table contains at most
$n\cdot(1+\tfrac{m}{\lambda})^{\lambda}$ entries, each updated over at most $(1+\tfrac{m}{\lambda})^{\lambda}$ admissible bundles. 
Combining these factors gives an overall running time of
$poly(n,m)\cdot \Big(1+\tfrac{m}{\lambda}\Big)^{2\lambda}$.

This is polynomial in the input size when $\lambda$ is constant. Hence, \localBudgetAbrv{} can be solved in polynomial time for a constant number of good types $\lambda$ and identical agent valuations.
\end{proof}

\subsection{When there are \(\tau\) agent types and \(\lambda \) good types}

We now extend the identical-valuation case to instances with $\tau$ agent types and $\lambda$ good types, covering settings with a small number of worker or machine classes and categories of interchangeable goods, such as shifts or equipment.

\begin{theorem}\label{thm:more-types}
Assuming integral payments, \localBudgetAbrv{} can be solved in
$\mathrm{poly}(n,m)\cdot (1 + n \cdot b_{\max})^{\tau}
\left(1+\frac{m}{\lambda}\right)^{\lambda(\tau+1)}$
time, where $b_{\max}:=\max_{a\in A} b_a$ is the maximum individual budget.
Thus, for constant numbers of agent types $\tau$ and good
types $\lambda$, the problem is solvable in pseudo-polynomial time.
\end{theorem}

Let \( I \) be an instance of the \localBudgetAbrv{} problem.  
We partition the set of agents into \(\tau\) disjoint types \(A_1, A_2, \dots, A_\tau\),  
such that all agents in \(A_t\) share the same valuation function \(v_t(\cdot)\).  
Let \(n_t = |A_t|\) denote the number of agents of type \(t\).

We now introduce notation that will be used throughout the proof.

Let $(\mathcal{X}^*, p^*)$ be an envy-free, budget-feasible allocation with integral payments.

    For each agent type \(t \in [\tau]\), we denote by \(P_t^*\) the \emph{total payment} made by all agents of that type in $(\mathcal{X}^*, p^*)$, i.e.,
    $P_t^* = \sum_{a \in A_t} p_a^*$.

    Let \(\Gamma_t^*\) denote the set of goods allocated to all agents of type \(t\), i.e.,
$\Gamma_t^* = \bigcup_{a \in A_t} X_a^*$.

For each good type \(j \in [\lambda]\), let \(m_{t,j}^*\) represent the total number of items of that type assigned to agents of type \(t\).
The corresponding item-count vector of \(\Gamma_t^*\) is then
\[
\mathbf{m}_t^* = (m_{t,1}^*,\, m_{t,2}^*,\, \dots,\, m_{t,\lambda}^*).
\]

    Because agents of the same type $t$ evaluate items and payments identically, an envy-free outcome requires them to receive equal utility. The resulting \emph{per-agent utility induced by the allocation and payments} is
     \[u_t^* = \frac{1}{n_t}\bigl(v_t(\Gamma_t^*) + P_t^*\bigr),\]
    which follows directly from the equality \(v_t(\Gamma_t^*) + P_t^* = n_t u_t^*\).
    
    Collecting these across all types, we define the \emph{utility vector}
    $\mathbf{u}^* = (u_1^*, u_2^*, \dots, u_\tau^*)$,
    which represents the per-type utilities induced by the envy-free, budget-feasible solution.

Given a utility vector, we define the notion of an admissible bundle set of an agent $a \in A_t$ as follows. 
\begin{definition}[Admissible bundle set]
\label{def:admissblebundle}
\begin{align*}
    \mathcal{K}_a(\mathbf{u}) \coloneqq 
    \Bigg\{ (k_1, \dots, k_\lambda) &\in \mathbb{Z}_{\ge 0}^\lambda\ :\ 
    \\[4pt]
    p_a &= u_t - \sum_{j=1}^\lambda k_j v_t(g_j)\ \in \mathbb{Z},
    \\[4pt]
    p_a &\ge -b_a,
    \\[4pt]
    \sum_{j=1}^\lambda k_j v_{t'}(g_j) + p_a &\le u_{t'}\ \ \forall\, t' \in [\tau],
    \\[4pt]
    0 &\le k_j \le m_{t,j}\ \ \forall\, j \in [\lambda]
    \Bigg\}.
\end{align*}
\end{definition}

Here:

\begin{itemize}
    \item The first condition ensures that the agent’s utility equals the target per-type utility \(u_t\), with a payment \(p_a\) consistent with the budget.
    \item The second condition enforces envy-freeness by making sure that the set is not valued higher by some agent type $t'$ than their target utility $u_{t'}$.
    \item The budget constraint \(p_a \ge -b_a\) guarantees that the payment is feasible for agent $a$.
    \item The final bound \(0 \le k_j \le m_{t,j}\) ensures that the allocation does not exceed the available number of goods of each type.

\end{itemize}

We are now ready to formally define our algorithm.
Given an instance of \localBudgetAbrv{} with $\tau$ types of agents and $\lambda$ types of goods, the procedure is as follows:

We first guess the total payments and goods allocated to each agent type,
enumerating all integer vectors of total payments
    \((P_1, \dots, P_\tau)\) such that 
    \(\sum_{t=1}^\tau P_t = 0\) and each 
    \(P_t\) is an integer from the interval \([-n_t B,\, (n-n_t) B]\), 
    where \(B\) is the maximum individual budget.  
    For each agent type \(t \in [\tau]\), we also guess a \emph{goods vector} 
    \(\Gamma_t = (m_{t,1}, m_{t,2}, \dots, m_{t,\lambda})\),
    where \(\ell_{t,j}\) denotes the number of goods of type \(g_j\) allocated in total to agents of type \(t\) and $\sum_{t \in [\tau]} m_{t,j} = m_j$ for all $g_j \in G$.
    For each agent type $t$, we compute the per-agent target utility
    \[
    u_t = \frac{1}{n_t}\bigl(v_t(\Gamma_t) + P_t\bigr).
    \]
    which defines the \emph{guessed utility vector} 
    $\mathbf{u} = (u_1, \dots, u_\tau)$.

    We can now construct a dynamic programming table $T_t$ for every agent type $t$ with entries 
$T_t[i, \ell_{t,1}, \dots, \ell_{t,\lambda}] \in \{0,1\}$
where $i \in [0..n_t]$ indexes agents of type $t$, and each $0 \leq \ell_{t,j} \leq m_{t,j}$ denotes the number of goods of type $j$ allocated so far. Set the base case
\[
T_t[0, 0, \dots, 0] = 1, \qquad 
T_t[0, \ell_{t,1}, \dots, \ell_{t,\lambda}] = 0 
\] for all other $(\ell_{t,1}, \dots, \ell_{t,\lambda})$.

 Next, for each agent $a_i \in A_t$, we define the admissible bundle set $\mathcal{K}_{a_i}(\mathbf{u})$ as per \Cref{def:admissblebundle}. Then, for each $i \in [n_t]$ and for every feasible goods vector $\mathbf{\ell}=(\ell_{t,1}, \dots, \ell_{t,\lambda})$:
    \[
    T_t[i, \mathbf{\ell}] =
    \bigvee_{(k_1,\dots,k_\lambda) \in \mathcal{K}_{a_i}(\mathbf{u})} 
    T_t[i-1, \ell_{t,1} - k_1, \dots, \ell_{t,\lambda} - k_\lambda].
    \]

    After filling the DP tables for all types, verify whether 
    \[
    T_t[n_t, m_{t,1}, \dots, m_{t,\lambda}] = 1 
    \quad \text{for all } t \in [\tau],
    \]
    where $(m_{t,1}, \dots, m_{t,\lambda})$ is the goods vector guessed for type $t$.  
    If this condition holds for every type, the corresponding allocation and payments constitute a feasible envy-free solution and are returned; otherwise, proceed to the next guess.

    If no combination of guesses yields a feasible DP solution, the algorithm concludes that the given instance is a \texttt{no}-instance.

\begin{lemma}
\label{lem:dp-accepts-correct-guesses}
Let \( I \) be an instance of the \localBudgetAbrv{} problem with \(\tau\) types of agents and \(\lambda\) types of goods.  
There exists an envy-free and budget-feasible solution \((\allocation^*, \mathbf{p}^*)\), that can be represented by \((P_t^*)_{t \in [\tau]}\) and \((\Gamma_t^*)_{t \in [\tau]}\), with integral payments,  
if and only if for all dynamic programming tables \((T_t)_{t \in [\tau]}\) under this pair, $T_t[n_t, m_{t,1}^*, \dots, m_{t,\lambda}^*] = 1$.
\end{lemma}

\begin{proof}
For the forward implication, suppose that $I$ is a \texttt{yes}-instance of the \localBudgetAbrv{} problem with a corresponding envy-free, budget-feasible solution $(\allocation^*, \textbf{p}^*)$. Let $(P_t^*)_{t \in [\tau]}$ and $(\Gamma_t^*)_{t \in [\tau]}$ denote the sum of payments and set of goods assigned to each type $t$ in this solution, and let $\mathbf{u}^* = (u_1^*, \dots, u_\tau^*)$ be the induced per-type utilities. By the properties of the solution, for each agent $a \in A_t$, the bundle assigned to $a$ belongs to the admissible set $\mathcal{K}_a(\mathbf{u}^*)$.

We prove the correctness of the DP table for an arbitrary type $t$.
Let $O_t[i,\ell_{t,1},\dots,\ell_{t,\lambda}]$ be a reference function that indicates whether the first $i$ agents of type~$t$ can be assigned admissible bundles 
that collectively use $\ell_{t,j}$ goods of type~$j$. 
Formally, we set $O_t[i,\ell_{t,1},\dots,\ell_{t,\lambda}] = 1$ if such an allocation exists, and $O_t[i,\ell_{t,1},\dots,\ell_{t,\lambda}] = 0$ otherwise. By construction of the solution, it follows that 
\[
O_t[n_t, m_{t,1}^*, \dots, m_{t,\lambda}^*] = 1.
\]

We prove by induction on $i$ that for all $(\ell_{t,1},\dots,\ell_{t,\lambda})$,  
\[
T_t[i,\ell_{t,1},\dots,\ell_{t,\lambda}] = O_t[i,\ell_{t,1},\dots,\ell_{t,\lambda}].
\]

\textbf{Base case ($i=0$).} The DP initializes with $T_t[0,0,\dots,0]=1$ and $T_t[0,\ell_1,\dots,\ell_\lambda]=0$ for all nonzero tuples. This exactly matches $O_t[0,\ell_1,\dots,\ell_\lambda]$, since no goods are assigned when there are zero agents.

\textbf{Induction step.} Assume the claim holds for $i-1$.  

 $T \implies O$: If $T_t[i,\ell_{t,1},\dots,\ell_{t,\lambda}] = 1$, then by the DP recurrence there exists $(k_1,\dots,k_\lambda)\in \mathcal{K}_{a_i}(\mathbf{u}^*)$ such that 
\[
T_t[i-1, \ell_{t,1}-k_1,\dots,\ell_{t,\lambda}-k_{\lambda}] = 1.
\]  
By induction, $O_t[i-1, \ell_{t,1}-k_1,\dots,\ell_{t,\lambda}-k_{\lambda}] = 1$, so appending $(k_1,\dots,k_\lambda)$ as the bundle for agent $i$ gives $O_t[i,\ell_{t,1},\dots,\ell_{t,\lambda}] = 1$.

 $T \impliedby O$: Conversely, if $O_t[i,\ell_{t,1},\dots,\ell_{t,\lambda}]=1$, then the last agent’s bundle $(k_1,\dots,k_\lambda)\in \mathcal{K}_{a_i}(\mathbf{u}^*)$ ensures that 
\[
O_t[i-1, \ell_{t,1}-k_1,\dots,\ell_{t,\lambda}-k_{\lambda}] = 1,
\]  
and by induction, $T_t[i-1, \ell_{t,1}-k_1,\dots,\ell_{t,\lambda}-k_{\lambda}] = 1$, so the DP recurrence sets $T_t[i,\ell_{t,1},\dots,\ell_{t,\lambda}] = 1$.
Hence, in particular,  
$T_t[n_t, m_{t,1}^*, \dots, m_{t,\lambda}^*] = 1$.

Since $t\in[\tau]$ was arbitrary, this proves the forward implication.

For the reverse implication, suppose that, for some guessed
$(P_t^*,\Gamma_t^*)_{t \in [\tau]}$ the DP entry $T_t[n_t, m_{t,1}^*, \dots, m_{t,\lambda}^*]=1$ for all $t\in [\tau]$. By the definition of the dynamic program, this means there exists an assignment of bundles from the admissible sets to all agents of type $t$ that exactly matches $(m_{t,1}^*, \dots, m_{t,\lambda}^*)$. By construction of the admissible sets, each assigned bundle yields the target utility for that agent type, respects the corresponding budget constraints, and ensures envy-freeness across all types. Applying this reasoning iteratively for every type $t \in [\tau]$, we obtain a global allocation and payment profile that jointly satisfy feasibility, budget balance, and envy-freeness. Therefore, the instance is a \texttt{yes}-instance.
\end{proof}

Using \cref{lem:dp-accepts-correct-guesses}, we can now prove \cref{thm:more-types}.

\begin{proof}[Proof of \cref{thm:more-types}]
The algorithm enumerates all integral guesses for payments $P_t$ and subsets $\Gamma_t$ for each type $t \in \tau$, for which we have budget balance ($\sum_{t\in \tau} P_t = 0$) and completeness ($\bigcup_{t\in [\tau]} \Gamma_t = G$). It then computes the corresponding utility vector $\mathbf{u}$, and the corresponding type-wise DPs $(T_t)_{t\in\tau}$ to check feasibility. By \cref{lem:dp-accepts-correct-guesses}, the DP returns $1$ for all types if and only if the instance is a \texttt{yes}-instance, ensuring correctness.

Conversely, by \cref{lem:dp-accepts-correct-guesses}, if for some guessed vector the DP entries satisfy $T_t[n_t,m_{t,1},\dots,m_{t,\lambda}]=1$ for all $t \in [\tau]$, then the backpointers yield bundles and payments that form a feasible envy-free, budget-balanced solution, so the algorithm accepts if and only if the instance is a \texttt{yes}-instance. 

\textbf{Running time.}
Let $b_{\max}$ denote the largest individual budget bound, and let $m_j$ be the number of goods of type~$j$ (so that $m = \sum_{j=1}^\lambda m_j$).  
Each total payment $P_t$ can take up to $1 + nb_{\max}$ integer values, hence there are at most $(1 + nb_{\max})^{\tau}$ possible payment guesses.

For the allocation, each good type $j$ must be divided among the $\tau$ agent types. Applying the AM--GM inequality to the terms $(1+m_1),\dots,(1+m_\lambda)$, we obtain the bound
\[
\!\big(1+\tfrac{m}{\lambda}\big)^{\lambda(\tau -1 )}.
\]
For a fixed global guess (payments and good counts), all DP tables together contain at most
$n\cdot(1+\tfrac{m}{\lambda})^{\lambda}$ entries, each updated over at most $(1+\tfrac{m}{\lambda})^{\lambda}$ admissible bundles.  
Thus the total DP cost per guess is $\mathcal{O}\big(n(1+\tfrac{m}{\lambda})^{2\lambda}\big)$.

Combining these factors gives an overall running time of
\[
poly(n,m)\cdot (1 +nb_{\max})^{\tau}\cdot \Big(1+\tfrac{m}{\lambda}\Big)^{\lambda(\tau+1)}.
\]
When $\tau$ and $\lambda$ are constants, this is pseudo-polynomial in the numerical parameters (budgets, valuations, and goods counts).
\end{proof}

\begin{remark}
                The same approach remains valid when payments are constrained to be integral across agent types but not necessarily for individual agents. The payment guesses still enforce type-level integrality, while the dynamic program permits fractional utilities within each type.
\end{remark}

This resolves the EF part of an open problem of~\citet{DBLP:conf/ijcai/0004R23}, who asked whether EF with few agent types and few item types can be solved in polynomial time. Since our model generalizes EF (corresponding to budget zero), Theorem~\ref{thm:more-types} implies a positive answer for explicitly listed items: when payments are absent, the only payment vector is the zero vector (hence $(1 + n \cdot b_{\max})^{\tau}=1$), and the running time becomes polynomial for constant $\tau$ and $\lambda$. The binary high-multiplicity variant remains open.

\section{Discussion}\label{conclusion}
In this paper, we studied envy-free allocations with monetary transfers under individual budget constraints. This extends the well-studied subsidies model
\cite{DBLP:conf/sagt/HalpernS19,SubsidiesGoko} by replacing external subsidies with budget-balanced transfers between agents.
Our tractability results for identical valuations and for instances with few agent and good types identify tractable structure within the general problem. Several natural questions remain open. In particular, it remains open whether \localBudgetAbrv{} is fixed-parameter tractable with respect to the number of
good types~$\lambda$ alone.

A particularly relevant direction for future work is the analysis of strategy-proofness. \citet{tadenuma1995games} show that no mechanism can simultaneously satisfy envy-freeness, strategy-proofness, and budget balance in a quasi-linear environment with at least three agents and additive valuations. Since their setting is a special case of our model when budgets are sufficiently large, this impossibility also applies to our framework. This motivates investigating whether strategy-proofness can be recovered by relaxing budget balance, for example by allowing external subsidies or donations, or by imposing structural restrictions on the budgets~\cite{strat}. Individual budgets may also create new manipulation opportunities, as agents can misreport both valuations and budgets.

While our focus is on goods, extending our model and algorithmic results to chores and mixed settings is a natural direction for future work. In such settings, agents who receive no chores may need to make payments
to maintain envy-freeness, creating new interactions between allocations
and individual budgets.

\section*{Acknowledgments}
This project was initiated at IIT Jodhpur during a research exchange supported by a DST–DAAD travel grant (DAAD Grant Agreement No.\ 57683502).  
Eva Deltl acknowledges support from the Deutsche Forschungsgemeinschaft (DFG, German Research Foundation), project COMSOC-MPMS (Grant Agreement No.\ 465371386).  
Pallavi Jain acknowledges support from the SERB-SUPRA grant (No.\ SPR/2021/000860) and the IIT Jodhpur Seed Grant (No.\ I/SEED/PJ/20210119).  
Tanmay Inamdar acknowledges support from the Anusandhan National Research Foundation (ANRF, Grant No.\ ANRF/ECRG/2024/004144/PMS) and the IIT Jodhpur Research Initiation Grant (Grant No.\ I/RIG/TNI/20240072). 


\bibliographystyle{abbrvnat}
\bibliography{arxiv}


\end{document}